\documentclass[a4paper,UKenglish,numberwithinsect,cleveref,thm-restate]{lipics-v2021}
\pdfoutput=1 % for arxiv
\usepackage{amsfonts}
\usepackage{tikz}
\usepackage{todonotes}

\title{(Sub)linear kernels for edge modification problems towards structured graph classes}
\author{Gabriel Bathie}{\'{E}cole Normale Sup\'{e}rieure de Lyon, France}{gabriel.bathie@ens-lyon.fr}{}{}
\author{Nicolas Bousquet}{LIRIS, CNRS, Universit\'e Claude Bernard Lyon 1, Université de Lyon, France}{nicolas.bousquet@univ-lyon1.fr}{}{}
\author{Théo Pierron}{LIRIS, CNRS, Universit\'e Claude Bernard Lyon 1, Université de Lyon, France}{theo.pierron@univ-lyon1.fr}{}{}

\funding{This work was supported by ANR project GrR (ANR-18-CE40-0032).}
\titlerunning{(Sub)linear kernels for edge modification problems}
\authorrunning{G. Bathie, N. Bousquet and T. Pierron}
\Copyright{Gabriel Bathie, Nicolas Bousquet and Théo Pierron}
\ccsdesc[500]{Theory of computation~Parameterized complexity and exact algorithms}
\keywords{kernelization, graph editing, split graphs, (sub)linear kernels}
\acknowledgements{}
\EventEditors{Petr A. Golovach and Meirav Zehavi}
\EventNoEds{2}
\EventLongTitle{16th International Symposium on Parameterized and Exact Computation (IPEC 2021)}
\EventShortTitle{IPEC 2021}
\EventAcronym{IPEC}
\EventYear{2021}
\EventDate{September 8--10, 2021}
\EventLocation{Lisbon, Portugal}
\EventLogo{}
\SeriesVolume{214}
\ArticleNo{8}
\hideLIPIcs
\nolinenumbers

\theoremstyle{plain}

\newtheorem{reductionrule}[theorem]{Rule} % Use this one if we want to have unique numbering, i.e not overlapping with lemmas and theorems.
\theoremstyle{definition}
\newtheorem{problem}{Problem}
\newtheorem{openproblem}{Open Problem}

\newcommand{\staradd}{\textsc{Starforest Addition}}%
\newcommand{\stardel}{\textsc{Starforest Deletion}}%
\newcommand{\staredit}{\textsc{Starforest Edition}}%
\newcommand{\tpadd}{\textsc{Trivially Perfect Addition}}%
\newcommand{\tpdel}{\textsc{Trivially Perfect Deletion}}%
\newcommand{\tpedit}{\textsc{Trivially Perfect Edition}}%
\newcommand{\splitadd}{\textsc{Split Addition}}%
\newcommand{\splitdel}{\textsc{Split Deletion}}%
\newcommand{\splitedit}{\textsc{Split Edition}}%
\newcommand{\cliqueadd}{\textsc{Clique + IS Addition}}%
\newcommand{\cliquedel}{\textsc{Clique + IS Deletion}}%
\newcommand{\cliqueedit}{\textsc{Clique + IS Edition}}%
\newcommand{\clusteredit}{\textsc{Cluster Editing}}%
\newcommand{\clusterdel}{\textsc{Cluster Deletion}}%
\newcommand{\vc}{\textsc{Vertex Cover}}%
\newcommand{\lnr}{\textsc{Label-And-Reduce}}%
\newcommand{\NN}{\mathbb{N}}%
\newcommand{\Sbb}{\mathbb{S}}%
\newcommand{\Gg}{\mathcal{G}}%
\newcommand{\NP}{\textsf{NP}}%
\DeclareMathOperator*{\poly}{poly}

\begin{document}

\maketitle

\begin{abstract}
In a (parameterized) graph edge modification problem, we are given a graph $G$, an integer $k$
and a (usually well-structured) class of graphs $\mathcal{G}$, and ask whether it is possible
to transform $G$ into a graph $G' \in \mathcal{G}$ by adding and/or removing
at most $k$ edges.
Parameterized graph edge modification problems received considerable attention in the last decades.

In this paper, we focus on finding small kernels for edge modification problems. One of the most studied problems 
is the \textsc{Cluster Editing} problem, in which the goal is to partition the vertex set into a disjoint union of cliques. 
Even if this problem admits a $2k$ kernel~\cite{cao2012cluster}, this kernel does not reduce the size of most instances. 
Therefore, we explore the question of whether linear kernels are a theoretical limit in edge modification problems, 
in particular when the target graphs are very structured (such as a partition into cliques for instance).
We prove, as far as we know, the first sublinear kernel for an edge modification problem. 
Namely, we show that \textsc{Clique + Independent Set Deletion}, which is a restriction of \textsc{Cluster Deletion}, 
admits a kernel of size $O(k/\log k)$. 

We also obtain small kernels for several other edge modification problems.
We prove that \textsc{Split Addition} (and the equivalent \textsc{Split Deletion}) admits a linear kernel, improving the existing quadratic kernel of Ghosh et al.~\cite{ghosh2015faster}. 
We complement this result by proving that \textsc{Trivially Perfect Addition}
admits a quadratic kernel (improving the cubic kernel of Guo~\cite{guo2007problem}), 
and finally prove that its triangle-free version (\textsc{Starforest Deletion}) 
admits a linear kernel, which is optimal under ETH.
\end{abstract}

\newpage

\section{Introduction}

A central problem in the context of data transmission, collection or storage,
is to recover the original information when the data has been altered.
Although it is not possible to know what the original data was in the general setting, 
it may be possible when we have some knowledge of the structure of the original data.
When we know that the alteration is limited, it is reasonable to assume 
that the original data is an element that has the desired structure and is the closest
to the altered data.

When the data that we are reconstructing is a graph, the problem becomes the following:
given a graph $G$ (the altered data) and a class of graphs $\Gg$ (the structure of the data), find
the graph in $\Gg$ that is the ``closest'' to $G$ (candidate for the original data).
There are multiple ways to define the distance between two graphs, but the most widely used
is the minimum number of vertex modifications or edge modifications needed to turn one into the other.
This type of problem, called graph modification problems, received considerable attention, 
for instance in computational biology~\cite{ben1999clustering}, machine learning~\cite{bansal2004correlation},
and image processing~\cite{wu1993optimal}.

In this work, we focus on edge modification problems, i.e. the distance is the minimum number of edge modifications (see Section~\ref{sec:prelim} for formal definitions of these problems).
A line of work, initiated by Yannakakis~\cite{yannakakis1978node}, showed that deciding 
whether a graph $G$ is at distance at most $k$ from $\Gg$ is \NP-complete for most classes of graphs,
even for very restricted classes such as bipartite graphs.
See~\cite{burzyn2006np,mancini2008graph,natanzon2001complexity} for an overview of the different results.

Therefore, in the last decades, edge modification problems received considerable attention from the point of view of parameterized complexity, 
which studies the resources required to solve \NP-complete problems in a fine-grained way.
See for example~\cite{bliznets2016lower,bocker2012golden,cai1996fixed,damaschke2014editing,drange2015fast,ghosh2015faster,van2018parameterizing}, and see~\cite{crespelle2020survey} for a recent survey on the topic.
In this paper, we will consider problems parameterized by the size $k$ of the solution (that is, the set of edges to add or remove).
%Note that in some specific cases, graph edition problems have been studied 
%with a different parameterization.
%A common additional parameter is the number of connected components in the target graph, 
%see for example~\cite{fomin2014tight} or~\cite{drange2015fast}. 

In this work, we focus on graph classes that can be characterized by a finite number of forbidden induced subgraphs. 
In his seminal paper, Cai~\cite{cai1996fixed} showed that, for every class of graphs $\Gg$ 
that can be characterized by a finite number of forbidden induced subgraphs, 
the $\Gg$-edge modifications problems are FPT parameterized by $k$. 
In other words, there exists a constant $c$, a function $f$ (that only depend on $\Gg$),
and an algorithm running in time $f(k) \cdot n^c$ that either finds a solution of size at most $k$, or returns that there is no such solution.
Therefore, most of the subsequent efforts focused
on determining for which $\Gg$ these problems admit a polynomial kernel.
Intuitively, a kernel is a polynomial-time preprocessing algorithm that extracts the ``hard'' part of an instance $(G,k)$: it solves easy parts of the instance and returns an equivalent instance $(G',k')$, whose size is bounded by $f(k)$, for some function $f$.
A kernel is a polynomial kernel if $f$ is a polynomial. 
% More formally, a problem $\Pi$ admits a \emph{polynomial kernel} if any instance $I$ of the problem $\Pi$ can be transformed into an equivalent instance of $\Pi$ of size at most $P(k)$ in polynomial time, where $P$ is a polynomial function. Kernels are efficient preprocessing algorithms,
% as they return smaller equivalent instances. 
The interested reader is referred to~\cite{fomin2019kernelization} for more details.

One of the most studied edge modification problem is 
\textsc{Cluster Editing}, in which the goal is to partition the graph into a disjoint union of cliques. 
This problem is known to admit a kernel with at most $2k$ vertices.
While this result seems impressive at first glance (for many parameterized problems, a linear kernel is asymptotically optimal), 
we can remark that here, we are comparing the number of vertices of the kernel with the number of edges in the solution.
It turns out that for most graphs in practice, the number of edges that have to be modified to obtain a cluster graph is larger than the number of vertices. 
For example, it is the case for most of the
public instances of the PACE challenge 2021 on cluster~editing\footnote{See \url{https://pacechallenge.org/2021/} for more information.}. 

This raises the question of whether linear kernels are optimal, in particular for \clusteredit{}.
We partially answer this question by giving a sublinear kernel for the closely related \cliquedel{} problem. It provides a ``proof of concept'' that linear kernels are not always optimal. As far as we know, it is the first example of a sublinear kernel for a graph edge modification problem. We complete this result with linear or quadratic kernels for several other edge modification problems.

Due to space constraints, all the proofs that are not in this extended abstract can be found in the appendix.

\subparagraph*{Our results.}

In this work, our goal is to understand when it is possible to obtain small, 
and in particular linear or sublinear kernels for edge modification problems. 
We focus in particular on graph classes where the vertex set can be partitioned into highly structured classes such as cliques or independent sets.
A typical example of such graph class is the class of split graphs, i.e. graphs that can
be partitioned into a clique and an independent set.

Most of our results are based on a high-level technique, that we call \lnr{}, 
which helps to design efficient kernelization algorithms for edge modification problems.
The key idea is to use the strong structure of each of the graphs in $\Gg$ to find a highly structured
partition $X_1, \dots, X_\ell$ of the graphs of $\Gg$ 
(e.g. a partition in cliques or independent sets, complete bipartition between subsets...).
We then define rules that label vertices $x$ as belonging to $X_i$, 
in such a way that if there is a solution, there is one for which $x \in X_i$.
We finally show that 1) when no rule can be applied, the number of unlabeled vertices is $O(\poly(k))$ when 
$(G,k)$ is a positive instance and 2) the number of labeled vertices in each class $X_i$
can be reduced to $O(\poly(k))$.

%%%%%%%%%%%%%%%%%%%%%%%%%%%%%%%%%%%%%%%%%%%%%%%
%%%%%%%%%%%%%%%%%%%%%%%%%%%%%%%%%%%%%%%%%%%%%%%
%\vspace{-.5cm}
\subparagraph*{Sublinear kernel for Clique+Independent Set deletion.}

The problems of graph edge modification towards cluster graphs have received considerable attention in the last
two decades in parameterized complexity, see e.g.~\cite{bocker2012golden,cao2012cluster,chen20122k, fomin2014tight, van2018parameterizing}. 
%The best kernel has currently size~$2k$~\cite{cao2012cluster,chen20122k}.

The deletion version, 
\clusterdel{} admits a cubic kernel~\cite{gramm2005graph}. We focus on a restricted version of this problem, where all clusters but at most one have size $1$. It corresponds to graphs that are the disjoint union of a clique and an independent set. In what follows, we will refer to this class as the class of \emph{clique + IS graphs}, 
and to the corresponding problem as \cliquedel{}. 
Since clique + IS graphs are $(P_3, 2K_2)$-free graphs, \cliquedel{} is FPT by~\cite{cai1996fixed}.

While \cliqueadd{} is trivial, both \cliquedel{} and \cliqueedit{} are \NP-complete 
(reduction from the \textsc{Clique} problem), 
and both can be solved in subexponential time 
($O^*(1.64^{\sqrt{k \ln k}})$ and $O^*(2^{\sqrt{k \ln k}})$ respectively\footnote{Recall that $O^*$ denotes the complexity up to polynomial factors.})~\cite{damaschke2014editing}.
Both problems also admit a simple $2k$-kernel,
based on twin reduction rules~\cite{crespelle2020survey}.

Our result, proved in Section~\ref{sec:cliqueis}, is the following.
\begin{restatable}[]{theorem}{cliqueis}\label{thm:clique_is}
	\cliquedel{} admits a kernel of size $2k/\log k + 1$.
\end{restatable}

Our algorithm uses the structure of clique+IS graphs to remove vertices with small degree 
and to reduce the instance when the minimum degree of the input is large. 
Theorem~\ref{thm:clique_is} is, as far as we know, the first sublinear kernel for edge modification problems. We conjecture that the size of this kernel is not optimal, and ask the following:

\begin{openproblem}
  Is there an $O(k^{1-\varepsilon})$ kernel for \cliquedel{} for some $\varepsilon > 0$?
\end{openproblem}

Moreover, it is plausible that other edge modification problems towards highly structured classes also admit a sublinear kernel.
A natural candidate is the closely related \clusteredit{} problem, which already admits a $2k$ kernel~\cite{cao2012cluster,chen20122k}.

\begin{openproblem}
	Does \clusteredit{} (resp. \clusterdel{}) admit a sublinear kernel? 
\end{openproblem} 

%%%%%%%%%%%%%%%%%%%%%%%%%%%%%%%%%%%%%%%%%%%%%%%
%%%%%%%%%%%%%%%%%%%%%%%%%%%%%%%%%%%%%%%%%%%%%%%
%\vspace{-.5cm}
\subparagraph*{Linear kernel for Split addition.}

Split graphs are graphs whose vertex set can be partitioned into a clique $K$ and an independent set $I$ 
(with no constraint on the set of edges between $K$ and $I$).
Since split graphs are auto-complementary, the \splitadd{} and \splitdel{} problems are equivalent.
Natanzon et al. showed that these two problems are \NP-complete~\cite{natanzon2001complexity}.
Since split graphs are $(2K_2,P_4,C_5)$-free graphs, the latter problems are FPT~\cite{cai1996fixed}.
Ghosh et al.~\cite{ghosh2015faster} later showed that these problems can be solved in subexponential $O^*(2^{O(\sqrt{k}\log k)})$ time
and that they admit a quadratic kernel.
Cygan et al.~\cite{cygan2015parameterized} improved the complexity to $O^*(2^{\sqrt{k}})$. 
Hammer and Simeone~\cite{hammer1981splittance} showed that, rather surprisingly, the related \splitedit{} problem
can be decided in polynomial time.

We improve upon the result of Ghosh et al.~\cite{ghosh2015faster} by showing that the \splitadd{} (and therefore \splitdel{}) problem admits a linear kernel in Section~\ref{sec:split}. 

\begin{restatable}[]{theorem}{split}\label{thm:split}
\splitadd{} and \splitdel{} admit a kernel with at most $11k + 6\sqrt{2k}+ 4$ vertices.
\end{restatable}

This result is the main technical contribution of the paper. 
From a very high-level perspective, our algorithm works as follows.
Let $(G,k)$ be a positive instance. If the clique of the solution is large enough, the neighborhood of many vertices of that clique has not been modified and we show that we can detect some of them and label them as clique vertices. Since the number of unlabeled clique vertices of the solution is bounded by a linear function, we can prove via a tricky and short argument that the number of unlabeled vertices of the independent set can be bounded. 
While the reduction rules are not very complicated, showing that the answer is negative when the number of unlabeled vertices is too large is the core of the proof.
We finally show that we reduce the number of labeled vertices to $O(k)$ vertices.

%%%%%%%%%%%%%%%%%%%%%%%%%%%%%%%%%%%%%%%%%%%%%%%
%%%%%%%%%%%%%%%%%%%%%%%%%%%%%%%%%%%%%%%%%%%%%%%
%\vspace{-.5cm}
\subparagraph*{Quadratic kernel for trivially perfect graphs.}

%The family of edge modification problems towards trivially perfect graphs received considerable attention in the last two decades. 
A trivially perfect graph is a graph such that for any pair	of adjacent vertices $u,v$, $N(u) \subseteq N(v)$ or $N(v) \subseteq N(u)$. 
The class of trivially perfect graphs can equivalently be characterized as the class of $(P_4,C_4)$-free graphs.
Drange et al.~\cite{drange2015exploring,drange2018polynomial} showed that, under the Exponential Time Hypothesis (ETH), \tpdel{} and \tpedit{} cannot be solved in subexponential time. 
Liu et al.~\cite{liu2015edge} gave an FPT algorithm for \tpdel{} running in time $O^*(2.42^{k})$.
On the other hand, the \tpadd{} problem does not admit such lower bounds. Drange et al.~\cite{drange2015exploring} designed a subexponential $O(2^{\sqrt{k \log k}}))$ 
algorithm for the problem and Bliznets et al.~\cite{bliznets2016lower} 
showed that assuming ETH, this cannot be improved beyond $O(2^{k^{1/4}})$.
In 2018, Drange and Pilipczuk~\cite{drange2018polynomial} 
showed that the three problems admit a polynomial kernel
of size $O(k^7)$, recently improved by Dumas et al.~\cite{DumasPT21+} into $O(k^3)$. 

In the specific case of \tpadd{}, a cubic kernel was already provided by Guo~\cite{guo2007problem}. 
We improve this result in Section~\ref{sec:trivially} by showing the following.

\begin{restatable}[]{theorem}{tpaddthm}\label{thm:tpadd}
	\tpadd{} admits a kernel with $2k^2+2k$ vertices.
\end{restatable}

Our kernel is based on the claim of Guo~\cite{guo2007problem}, 
which states that the instance can be reduced to vertices that belong
to at least one obstruction (that is, an induced $P_4$ or $C_4$).
Using this claim, Guo proved the existence of a cubic kernel.
By counting obstructions more precisely, we actually show very simply that the size of the kernel can be reduced to $O(k^2)$.  

%%%%%%%%%%%%%%%%%%%%%%%%%%%%%%%%%%%%%%%%%%%%%%%
%%%%%%%%%%%%%%%%%%%%%%%%%%%%%%%%%%%%%%%%%%%%%%%
%\vspace{-.5cm}
\subparagraph*{Linear kernelization of starforests}

We finally focus on triangle-free trivially perfect graphs, also known as star forests.
A \emph{star} is a tree with at most one internal vertex.
A star with $n$ vertices is called an \emph{$n$-star}. 
Note that the single vertex graph and $K_2$ are stars.
The class of \emph{starforests} graphs is the class of 
graphs that are a disjoint union of stars, that is every connected component is a star.
Alternatively, it may be defined as the class of $K_3,C_4,P_4$-free graphs.
%the class of triangle-free trivially perfect graphs,
%or the class of graphs with branchwidth 1~\cite{robertson1991graph}.

One can remark that removing an edge from a starforest yields another starforest,
hence it is never interesting to add edges to obtain a star forest. 
Therefore, \staradd{} is trivial, and \staredit{} is equivalent to \stardel{}.  
Drange et al. showed in~\cite{drange2015fast} that \stardel{}
is \NP-complete and cannot be solved in subexponential time (that is in time $O(2^{o(k)}\poly(n))$), assuming the ETH~\cite{impagliazzo2001complexity}. 

In Section~\ref{sec:star}, we prove the following result.

\begin{restatable}[]{theorem}{star}\label{thm:star}
	\stardel{} admits a kernel with at most $4k + 2$ vertices.
\end{restatable}

We also show that, under ETH, \stardel{} does not admit a sublinear kernel.
To the best of our knowledge, this work is the first formally published work on kernelization of \stardel{}.

\subparagraph*{Note.} Cao and Yuping~\cite{cao2021improved} obtained independently at the same time 
results that are very similar to ours: they designed the same quadratic kernel for \tpadd{} and obtained a similar kernel for \splitadd{}. However, they were only able to prove an $O(k^{1.5})$ upper bound for the latter, whereas we prove a tighter $O(k)$ bound.

\section{Preliminaries}
\label{sec:prelim}
\subparagraph*{Elementary definitions.}
In this work, all the graphs are undirected and simple (i.e. with no parallel edges or self-loops).
When $G$ is a graph, $V(G)$ denotes the set of vertices of $G$, and $E(G)$ denotes its set of edges.
Throughout the paper, we use $n$ (resp. $m$) to denote the size of $V(G)$ (resp. $E(G)$).
If $uv \in E(G)$, we say that $u$ and $v$ are \emph{adjacent}.
Given a vertex $u\in V(G)$, $N(u) = \{v$ such that $uv$ is an edge$\}$ is the \emph{open neighborhood} of $u$,
and $N[u] = N(u) \cup\{u\}$ is the \emph{closed neighborhood} of $u$. 
The \emph{degree} of $u$ in $G$, denoted $d(u)$, is the size of $N(u)$.
We use $\delta(G)$ to denote the minimum degree of $G$.
The \emph{complement graph} $\bar{G}$ of $G$ is the graph with vertex set $V(G)$ 
and edge set $\{uv : u \neq v \text{ and } uv \notin E(G)\}$.
A \emph{dominating set} of $G$ is a set $D$ of vertices of $G$
such that every vertex of $G$ is either in $D$ or adjacent to a vertex of $D$.
An \emph{independent set} of $G$ is a set $I$ of vertices of $G$
that are pairwise not adjacent.

%\vspace{-.5cm}
\subparagraph*{Kernelization algorithms.}

A \emph{kernelization algorithm} (in short, a kernel) is a polynomial-time algorithm 
that takes as input an instance $(G,k)$ of a parameterized problem $\Pi$ 
and outputs an instance $(G',k')$ that is positive if and only if $(G,k)$ is positive, the size of $G'$ is at most $f(k')$ for some computable function $f$. When $f$ is a polynomial, we say that the algorithm is a \emph{polynomial kernel}.
When dealing with graph problems, the size of the instance is often measured in terms of the number of vertices of $G'$.
Most kernelization algorithms (including those presented in this report) 
consist of the iterative application of \emph{reduction rules}. 
A reduction rule is a polynomial-time algorithm that input an instance $(G,k)$ and outputs another instance $(G',k')$.
We say that a reduction rule $R$ is \emph{safe} when $(G',k')$ is positive if and only if $(G,k)$ is.

%\vspace{-.5cm}
\subparagraph*{Graph edge modification problems.}

Let $\Gg$ be a class of graphs.
In a (parameterized) $\Gg$-graph edge modification problem, we are given a graph $G$, an integer $k$, and ask whether it is possible
to transform $G$ into a graph $G' \in \Gg$ by modifying 
(\textit{adding}, \textit{removing}, or doing both, which is called \textit{editing}) 
at most $k$ edges.

Given a set of edges $F$, we use the notation
$G+F$, $G-F$, and $G \Delta F$ to denote the graphs with vertex set $V(G)$ and respective set of edges
$E(G) \cup F$, $E(G) \setminus F$ and $E(G) \Delta F$.

Formally, we will consider the following problems:
\begin{problem}[$\mathcal{G}$-\textsc{Addition} (resp. \textsc{Deletion}, resp. \textsc{Edition})]$ $\newline%Hack to start a new line
	\textbf{Input}: A graph $G$, an integer $k\in\NN$.\\
	\textbf{Output}: ``YES'' if there exists a set $F$ of at most $k$ edges of $G$
	such that $G+F$ (resp. $G-F$, resp. $G \Delta F$) is in $\mathcal{G}$, ``NO'' otherwise.
\end{problem}

% Here, we focus on cases where $\Gg$ is characterized by 
% a finite set of forbidden induced subgraphs. 
% An example of such graphs is the class of split graphs, which are $2K_2,P_4,C_5$-free graphs.
% Hence we only consider hereditary graph classes $\Gg$, that is, if $G \in \Gg$,
% then any induced subgraph $H$ of $G$ is also in $\Gg$.
% Since the operations of adding edges and taking an induced subgraph commute,
% we obtain that if $(G,k)$ is a positive instance of a $\Gg$-edge modification problem,
% then for any subgraph $H$ of $G$, $(H,k)$ is also a positive instance.% of the same problem. 

\section{Sublinear kernel for the Clique + Independent set deletion problem}\label{sec:cliqueis}

The goal of this section is to prove Theorem~\ref{thm:clique_is}. To obtain the announced kernel, we apply the \lnr{} technique.
For this problem, the labeling rules aim to identify vertices that will be in the independent set of a solution if it exists.
We can then delete all the edges incident to these vertices, decrease the parameter accordingly and remove these vertices from the graph.

We assume that $k$ is smaller than $m$ since otherwise, 
the instance is trivial: we obtain an independent set (which is a clique+IS graph) by deleting all the edges in $G$.

\begin{reductionrule}[Low degree reduction rule 1]\label{clique-lowdegree}
	If $v \in V(G)$ has degree $d(v) < \sqrt{2(m-k)} - 1$,
	delete $v$ from $G$ and decrease the parameter by $d(v)$.
\end{reductionrule}

This rule can be implemented to run in linear time. It is
moreover safe. Indeed, since we consider the deletion problem, any
vertex $v$ deleted by the rule has degree smaller than
$\sqrt{2(m-k)} - 1$ in $G-F$, hence cannot be in the clique of an
optimal solution according to the following lemma.

\begin{lemma}\label{lemma:lb-clique}
	Let $(G,k)$ is a positive instance of \cliquedel{}. If $F$ is a solution of $(G,k)$, then the clique in $G-F$ has size at least $\sqrt{2(m-k)}$.
\end{lemma}
\begin{proof}
	Since $F$ contains at most $k$ edges, the graph $G-F$
	has at least $m-k$ edges. 
	Moreover, $G-F$ is a clique+IS graph, therefore
	all its edges are the edges of a unique clique. 
	Therefore, the size $c$ of the clique of $G-F$ satisfies ${c \choose 2}\geqslant m-k$, hence $c\geqslant \sqrt{2(m-k)}$.
\end{proof}

One can prove that this first rule can be extended to obtain a
linear kernel: when this rule cannot be applied, assuming that $m\geq 2k$ implies that $n = O(\sqrt{m})$.
When $m = O(k^2)$, we are done, and when $m = \Omega(k^2)$,
the minimum degree of the graph is $\Omega(k)$, therefore at most $O(1)$ vertices can be removed,
and in that case, the existence of a solution can be tested in polynomial time. 

To further reduce the size of the kernel to
$O(k/\log k)$, we use the two following rules to take care of very
sparse or very dense instances.

\begin{reductionrule}[Low degree reduction rule 2]\label{clique-logdegree}
	Let $v$ be a vertex of degree at most $2\log k-1$.
	If there is no solution $F$ of $(G,k)$ such that 
	$v$ is in the clique of $G-F$,
	remove $v$ from $G$ and decrease $k$ by $d(v)$.
\end{reductionrule}

This rule is trivially safe.
Moreover, it can be performed in polynomial time. Indeed, since
we consider an edge-deletion problem, if $v$ lies in the clique $K$ of
$G-F$, then every vertex of $K$ is adjacent to $v$ in $G$,
i.e. $K\subseteq N[v]$.  Since the degree of $v$ is at most
$2\log k - 1$, there are at most $k^2 $ subsets in $N[v]$. We can
therefore try all of them and decide in polynomial time whether there
exists a solution $F$ of $(G,k)$ such that $v$ is in the clique in
$G-F$.

\begin{reductionrule}[High degree]\label{clique-highdensity}
	If $G$ has minimum degree $\delta(G) \geq k/(2\log k)$,
	solve the instance and output a trivial equivalent instance.
\end{reductionrule}

Again, this rule is clearly safe. The not-so-easy part is to show that
\cliquedel{} can be decided in polynomial time when
$\delta(G) \geq k/(2\log k)$.

\begin{restatable}[]{lemma}{cliqueisolfirst}
	Rule~\ref{clique-highdensity} can be applied in polynomial time.
\end{restatable}

To finish the proof of Theorem~\ref{thm:clique_is}, it remains to bound
the size of the reduced graph. This is the goal of the following
lemma.

\begin{restatable}[]{lemma}{cliqueisolsecond}\label{lemma:positive-clique}
	If $(G,k)$ is a positive instance and none of the rules can be
	applied, then $|V(G)| \leq 2 \cdot \frac{k}{\log k} + 1$.
\end{restatable}

\subparagraph*{Concluding remarks.} 
% Kernel sizes (in terms of number of vertices) for graph edge modification problems usually come from an upper bound on the number of edges in a positive instance, which translates into a similar upper bound on the number of vertices.
% However, this is often far from tight. To obtain our sublinear kernel, we use the fact that
% a linear upper bound on the number of edges yields a sublinear upper bound on the number of vertices.

% For example, in clique graphs with few isolated vertices (e.g. $o(n)$ isolated vertices),
% we have $m = \Theta(n^2)$, hence such a graph with $O(f(k))$ edges has $O(\sqrt{f(k)})$ vertices. 
% Therefore, one could hope that \cliquedel{} admits a kernel of size $O(\sqrt{k})$.
% However, we were not able to obtain such a result. 
% Note that our reduction rules ensure that when the number of edges is far from $k$
% and the number of missing edges is far from $k$, we indeed have an $O(\sqrt{k})$-kernel. 
% Hence, in order to improve the size of the kernel,
% one only needs to take into account instances where the number of edges (or non-edges) is close to $k$.
% %We leave the following question as an open problem.
% %\begin{openproblem}
% %	  Is there a $O(k^{1-\varepsilon})$ kernel for \cliquedel{} for some $\varepsilon > 0$?
% %\end{openproblem}

Finally, one can easily show that Rule~\ref{clique-lowdegree} can be adapted for the \cliqueedit{} problem 
by modifying the constant. 
On the other hand, it seems that Rules~\ref{clique-logdegree} and~\ref{clique-highdensity} 
do not readily generalize to \cliqueedit{}, therefore we were not able to obtain an $O(k/\log k)$ kernel for this problem. 
However, it is an easy exercise to show that we can weaken them to obtain a kernel with at most $k/c$ vertices for any possible constant $c > 1$, at the cost of a running time in $O(n^c)$.

\section{Linear kernel for addition towards split graphs}\label{sec:split}

The goal of this section is to prove Theorem~\ref{thm:split}.
Since the class of split graphs is closed under complementation,
it is sufficient to prove that Theorem~\ref{thm:split} 
holds for \splitadd{}.

% We use the \lnr{} technique for the problem \splitadd{}.
We use the structure of the input graph to detect and label vertices 
that will be in the clique or the independent set part of a split decomposition 
of a well-chosen solution.
More precisely, we show that, if the instance $(G,k)$ is positive, the labeling constructed by our algorithm satisfies that
there exists a solution $F$ of $(G,k)$ and a split decomposition $(K^*, I^*)$ of $G+F$ 
such that all the vertices labeled as ``clique'' (resp. ``independent set'') 
are in $K^*$ (resp. $I^*$).
We then prove that if $(G,k)$ is a positive instance, 
then the number of unlabeled vertices at the end of the algorithm is $O(k)$.
Moreover, we show that we can reduce the number of labeled vertices to $O(k)$.
Combining the above yields a linear kernel.

We present our reduction rules in Section~\ref{sec:split_rules} and prove them in subsequent sections.

\subsection{Labeling and reduction rules}\label{sec:split_rules}

Our algorithm keeps track of a partition $(K,I,D)$ of $V(G)$, which
corresponds to the labels of the vertices of $G$. The set $K$
(resp. $I$) stands for the vertices already labeled ``clique''
(resp. ``independent set'') while $D$ (for ``don't know'') contains
the vertices that are not yet labeled.	Initially, no vertex is
labeled, hence $K = \emptyset, I = \emptyset$ and $D = V(G)$.

We will apply the following reduction rules, whose correction is postponed to Section~\ref{sec:split-safe}.

\begin{reductionrule}[I-rules]\label{split-is} 
	Move $v\in D$ to $I$ whenever at least one of the following holds:
	\begin{alphaenumerate}
		\item\label{split-is-r1} $v$ has all of its neighbors in $K$,
		\item\label{split-is-r2} $v$ is non-adjacent to at least $k+1$ vertices of $K$.
	\end{alphaenumerate}
\end{reductionrule}
Notice that this rule applies to isolated vertices since whenever $v$
is isolated, $N(v) = \emptyset \subseteq K$.

\begin{reductionrule}[K-rules]\label{split-k} 
	Move $v\in D$ to $K$ whenever at least one of the following holds:
	\begin{alphaenumerate}
		\item\label{split-k-r1} $v$ has a neighbor in $I$,
		\item\label{split-k-r2} $N(v)$ contains at least $k+1$ non-edges,
		\item\label{split-k-r3} $v$ dominates $K \cup D$.
	\end{alphaenumerate}
\end{reductionrule}

The following reduction rule simply ensures that $K$ is a clique and $I$ an independent set.

\begin{reductionrule}[Reduction rules]\label{split-red} 
	Apply one of the following rules as long as possible:
	\begin{alphaenumerate}
		\item\label{split-red-r1} if there is a non-edge $e$ between vertices of $K$, then 
		add $e$ to $E(G)$ and decrease $k$ by~1.
		\item\label{split-red-r2} if $k < 0$ then return a trivially negative instance.
		\item\label{split-red-r3} if there is an edge between vertices of $I$, then return a trivially negative instance.
	\end{alphaenumerate}
\end{reductionrule}

We apply these rules exhaustively, and stop when none can be applied. 
At each step, we remove a vertex from $D$ or we add an edge to $G$.
Then the algorithm stops after at most $n^2$ steps. 
Moreover, one can easily apply the rules in polynomial time.
When none of the previous rules can be applied, we apply the following reduction rule.

\begin{reductionrule}[Unlabeling algorithm]\label{split-unlabel}
	\begin{alphaenumerate}
	\item \label{split-unlabel-r2} If $K$ contains at least $k+1$ vertices, replace $K$ by a set 
	$K' = \{v_1',\ldots,v_k'\}$ of $k$ vertices and denote by $G'$ the resulting graph.
	Moreover, for each vertex $v\in D$, if $v$ is non-adjacent to $t$ vertices of $K$ in $G$, 
	then connect $v$ to $v_{t+1}', \ldots, v_k'$ and do not connect it to $v_1',\ldots,v_t'$.
	\item \label{split-unlabel-r3} Replace $I$ by an independent set $I'$ of size $\sqrt{2k}$ connected to $K'$ in a complete bipartite manner, and not connected to $D$.
	\end{alphaenumerate}
\end{reductionrule}

With this rule, we can bound the number of vertices of the resulting graph by $|D|$ plus at most $k$ vertices (for $K'$), plus at most $\sqrt{2k}$
vertices (for $I'$). Therefore, Theorem~\ref{thm:split} boils down to
the following lemma.

\begin{restatable}[]{lemma}{splitpos}\label{lemma:split-small-rest}
	If $(G,k)$ is a positive instance, then $|D| \leq 10k + 5\sqrt{2k}+ 4$.
\end{restatable}

While it is not very difficult to prove that the reduction rules are safe, the main technical contribution of this section consists in proving Lemma~\ref{lemma:split-small-rest}. 
The proof is split into two parts. First, we prove that the number of vertices of $D$ in the clique $K^*$ of the solution is linear in $k$. We prove it by showing that, if too many vertices of $K^*$ are in $D$, the neighborhood of many of them is not modified. And amongst them, one must be complete to $K\cup D$, a contradiction with Rule~\ref{split-k}.

Arguing that the number of vertices of $D$ in the independent set $I^*$ of the solution is $O(k)$ is more involved. 
First note that if a vertex has an independent set of size larger than $O(\sqrt{k})$ in its neighborhood, 
it is added to $K$ by Rule~\ref{split-k}-\ref{split-k-r2}. 
Since $D$ only contains $O(k)$ vertices in the clique, the number of vertices of $D$ 
in the independent set is at most $O(k^{3/2})$. 
To obtain a better upper bound on the size of $D$, 
we carefully distinguish the size of the neighborhood of the vertices of $D \cap K^*$ in $I^*$. Very roughly, we prove that the number of vertices in $D \cap K^*$ with many neighbors in $I^*$ is bounded by a sublinear function which permits to improve the size of the kernel. 
%And show that the number of vertices in the intersection of $D$ and the final independent set is $O(k)$.
The proof of Lemma~\ref{lemma:split-small-rest} is postponed to Section~\ref{sec:split-size}.

Lemma~\ref{lemma:split-small-rest} together with Rule~\ref{split-unlabel} ensure that the following reduction rule is correct, which completes the proof of Theorem~\ref{thm:split}:

\begin{reductionrule}[Final Rule]\label{split-reject}
If none of the previous rules can be applied, and the size of the instance is at least $11k + 6\sqrt{2k}+ 5$,
return a trivially negative instance.
\end{reductionrule}

%%%%%%%%%%%%%%%%%%%%%%%%%%%%%%%%%%%%%%%%%
%%%%%%%%%%%%%%% END OF DEF %%%%%%%%%%%%%%
%%%%%%%%%%%%%%%%%%%%%%%%%%%%%%%%%%%%%%%%%
\subsection{Correctness of the reduction rules}\label{sec:split-safe}

To analyze our algorithm,
we study the evolution of the instance $(G, k)$ with the partition $P = (K,I,D)$ after the application of each rule.
We will refer to the tuple $(G,k,P)$ as a \emph{generalized instance} of \splitadd{}. 
%At the beginning of the algorithm, we have $K = I = \emptyset$ and $D = V(G)$.

The following definition formalizes when a labeling of $G$ is compatible with a solution $F$.

\begin{definition}
	Let $H$ be a graph, let $F$ be a set of edges such that $H+F$ is a split graph, 
	and let $P = (K,I,D)$ be a partition of $V(H)$.
	We say that $P$ \emph{is compatible with} $F$, and denote it $P \vDash F$, 
	if there exists a split decomposition $(K^*, I^*)$ of $H+F$ such that $K\subseteq K^*$ and
	$I\subseteq I^*$. In that case, we say that the decomposition $(K^*, I^*)$ \emph{witnesses} the fact $P \vDash F$. 
\end{definition}

A generalized instance represents a graph along with a partial labeling of the vertices.
Such an instance is positive when there exists a solution that is compatible with the labeling.
This leads to the following definition.
\begin{definition}[Positive generalized instance]\label{split-gen-positive}
	A generalized instance $(G,k,P)$ is positive
	if there exists a solution $F$ of $(G,k)$ such that $P\vDash F$.
\end{definition}

This allows us to extend safeness properties to reduction rules operating on generalized instances.
We now show that the labeling and reduction rules preserve the existence of a solution.

\begin{restatable}[]{lemma}{splitfirst}\label{lm:split-compatible}
	Rules~\ref{split-is} to~\ref{split-red} are safe. 
\end{restatable}

Note that the initial labeling $P = (\emptyset, \emptyset, V(G))$ is
compatible with every solution of $(G,k)$ (if any). Therefore, by
applying transitively Lemma~\ref{lm:split-compatible}, we get that
%positive
% instances remain positive throughout the process. Moreover, %since the set
% %of vertices labeled ``clique'' (resp. ``independent set'') does not
% %decrease at each step,
% any solution for the reduced instance directly yields a solution for
% the initial one. Therefore,
the labeling and reduction process is
safe.
% positive (in the sense that it has a solution) if and only if the final generalized instance $(G',k',P')$
% is positive (in the sense of Definition~\ref{split-gen-positive}).

We finally show that Rule~\ref{split-unlabel} is safe. 

\begin{restatable}[]{lemma}{splitsecond}
	Rule~\ref{split-unlabel} is safe.
\end{restatable}

\subsection{Structure of positive instances}\label{sec:split-size}
% \section{Structure of positive Split Addition instances}\label{app:split}

This section is devoted to the proof of Lemma~\ref{lemma:split-small-rest},
restated below.%which states that, when the input $(G,k)$ is a positive instance,
%then the labeling/reduction process stops with at most $\leq 10k + 5\sqrt{2k}+ 4$
%unlabeled vertices.

\splitpos*

In what follows, we assume that the input is a positive instance 
and the labeling/reduction process stopped and returned a generalized instance $(G,k,P)$.
In particular, Rules~\ref{split-is} to~\ref{split-red} cannot be applied.
By Lemma~\ref{lm:split-compatible},
we get that there exists a solution $F$ of $(G,k)$
such that $|F|\leq k$ and $P\vDash F$.
Unrolling the definition, this means that there exists a split decomposition $(K^*, I^*)$ of $G + F$ 
such that $K\subseteq K^*$ and $I\subseteq I^*$. 
Let $K^D = D\cap K^*$ be the set of unlabeled vertices that belong to the clique,
and let $I^D = D\cap I^*$ be the set of the unlabeled vertices that belong to the independent set.
For every $v\in D$, let $I_v = N(v)\cap I^D$.
We give an upper bound on the cardinality of $D$ by
giving separate upper bounds on the respective cardinalities of $K^D$ and $I^D$.

Before diving into the details of the proof,
let us make two observations on the structure
of $D$, that follow from the fact that the labeling rules cannot be applied.

\begin{restatable}[]{observation}{firstobs}\label{obs:small_neighbor_k_i}
	For every vertex $v\in K^D$, $|I_v|\leq \sqrt{2k}+1$.
\end{restatable}

\begin{restatable}[]{observation}{secondobs}\label{obs:no_isolated}
	Every vertex $v\in I^D$ has a neighbor in $K^D$.
\end{restatable}

We first prove that $|K^D| = O(k)$. 

\begin{lemma}\label{claim:small_kd}
	We have $|K^D|\leq 4k$.
\end{lemma}
\begin{proof}
	Let us prove this statement by contradiction: 
	we prove that if $|K^D|\geq 4k+1$, then there is a vertex in $D$ that dominates
	$D\cup K$, which contradicts the fact that Rule~\ref{split-k}-\ref{split-k-r3}
	cannot be applied.
	
	By assumption, $K^*$ is a clique in $G+F$. Since $F$ contains at most $k$ edges, 
	there are at most $2k$ vertices of $K^D$ that are adjacent to edges of $F$.
	Since $|K^D|\geq 4k+1$, there are at least $2k+1$ vertices in $K^D$ that dominate
	$K^* = K^D \cup K$. Let $X$ denote the set of such vertices.
	We will now show that there is a vertex in $X$ that also dominates $I^D$,
	that is, a vertex of $K^D$ that dominates $K^D \cup I^D \cup K = D \cup K$.
	To prove the existence of this vertex, we will prove that for any vertex $u$
	in $X$ such that $I_u \neq I^D$, there exists a vertex $v\in X$
	such that $|I_v| > |I_u|$. By applying this property repeatedly, we eventually find a vertex
	$v$ such that $I_v = I^D$.

	Let $u$ be a vertex of $X$ such that $I_u \neq I^D$. 
	Since $K^D \subseteq N[u]$ and Rule~\ref{split-k}-\ref{split-k-r2} cannot be applied, 
	there are at most $k$ non-edges between $I_u$ and $K^D$.
	Hence, these non-edges are adjacent to at most $k$ vertices of $X$ ($X$ being a clique, every non-edge
	is incident to at most one vertex of $X$), and then
	at least $k+1$ vertices of $X$ dominate $I_u$. Let $X'$ be the subset of vertices of $X$ 
	that dominate $K^* \cup I_u$.
	Let $w$ be a vertex of $I^D \setminus I_u$.
	As noted in Observation~\ref{obs:no_isolated}, $w$ is adjacent to some vertex $v \in K^D$. Assume that $w$ is anticomplete to $X'$, so that $v\notin X'$. Since $v\in K^*$, every vertex of $X'$ is adjacent to $v$. Therefore $v$ contains at least $k+1$ non-edges in its neighborhood, namely the edges between $w$ and $X'$, a contradiction.

        Therefore, $v\in X'$ and the conclusion follows since $I_v$ contains $I_u$ and $w$. 
\end{proof}

By bounding locally the size of the neighborhood of each vertex in
$K^D$ using Observation~\ref{obs:small_neighbor_k_i}, Lemma~\ref{claim:small_kd}
directly provides an $O(k^{\frac32})$ kernel. However, as we will show,
this is not tight. Using a more global counting argument, we can show that
$|I^D| = O(k)$.

\begin{lemma}\label{claim:small_id}
	We have $|I^D| \leq 6k + 5\sqrt{2k}+ 4$.
\end{lemma}
\begin{proof}
	First, notice that Observation~\ref{obs:no_isolated} 
	implies that $I^D \subseteq \bigcup_{v\in K^D} N(v)$.	
	Therefore, if $|K^D|\leq \sqrt{8k}$, 
	Observation~\ref{obs:small_neighbor_k_i}
	implies the following upper bound on the cardinality of $I^D$:
	
	$$|I^D|\leq |K^D|\cdot(\sqrt{2k}+1) \leq 4k + 2\sqrt{2k} \leq 6k + 5\sqrt{2k} + 4.$$

	In what follows, we assume that $|K^D| > \sqrt{8k}$. We partition $I^D$ into two sets: $I^+$, the set of vertices 
	that have degree at least $|K^D|/4$, 
	i.e. vertices that are adjacent to at least $|K^D|/4$ vertices of $K^D$,
	and $I^- = I^D \setminus I^+$. We bound their sizes independently.

	First, by counting the number $n_e$ of edges between $K^D$ and $I^+$ 
	from the point of view of $K^D$, we get $n_e \leq |K^D|\cdot(\sqrt{2k}+1)$. 
	From the point of view of $I^+$, we get $n_e \geq |K^D|\cdot|I^+|/4$.
	By combining the two inequalities, we get $|I^+| \leq 4(\sqrt{2k}+1)$. 

        It remains to show that $|I^-|\leqslant 6k+\sqrt{2k}$. To this
        end, we consider two types of vertices in $K^D$: those that
        are adjacent to more than $\sqrt{2k}$ edges of $F$ in the solution, and the
        others. We then bound the number of vertices in $I^-$ adjacent
        to (at least) a vertex of each type.

        Since we add at most $k$ edges to $G$, there are at most
        $\sqrt{2k}$ vertices in $K^D$ incident to more than $\sqrt{2k}$ edges of $F$.  
        By Observation~\ref{obs:small_neighbor_k_i}, these vertices of $K^D$ have
        at most $\sqrt{2k}(\sqrt{2k}+1) \leq 2k + \sqrt{2k}$ neighbors
        in $I^D$ (and therefore in $I^-$).

        To conclude the proof, it is thus sufficient to show that there are
        at most $4k$ vertices in $I^-$ that are adjacent to vertices of $K^D$
        of the second type.

        Let $v$ be a vertex of $K^D$ of the second type. We write
        $K_v=N(v)\cap K^D$ and $I^-_v = N(v)\cap I^-$. Observe that,
        by definition, $|K_v|\geqslant |K^D| - \sqrt{2k}
        \geq|K^D|/2$. Let $\bar{d}$ be the average degree in
        $K_v$ of vertices in $I^-_v$. Since
        Rule~\ref{split-k}-\ref{split-k-r2} cannot be applied, there
        are at least $|K_v|\cdot|I^-_v| - k$ edges between $K_v$ and
        $I^-_v$, hence $\bar{d}\geqslant |K_v|-k/|I^-_v|\geqslant |K^D|/2 - k
        /|I^-_v|$. However, by definition of $I^-$, each vertex has
        degree at most $|K^D|/4$ in $K_v$ hence $\bar{d}\leqslant |K^D|/4$. Combining the above yields
        $|I^-_v|\leqslant 4k/|K^D|$. Since there are at most $|K^D|$
        vertices of the second type, the union of their neighborhoods has size at most $|K^D|\cdot 4k/|K^D|=4k$, which is the sought result.
\end{proof}

\section{Quadratic kernel for addition towards trivially perfect graphs}\label{sec:trivially}

The goal of this section is to prove Theorem~\ref{thm:tpadd}.
Recall that trivially perfect graphs are $(C_4,P_4)$-free graphs.
In what follows, we refer to induced $P_4$ or $C_4$ of a graph as its \emph{obstructions}. We say that a pair $(u,v)$ of vertices is a \emph{diagonal} if $uv \notin E$ and there exists two vertices $a,b$ such that $uavb$ is a $P_4$ or a $C_4$.
Given a diagonal $(u,v)$, the number of obstructions containing $(u,v)$ is the number of distinct pairs $a,b$ such that $uavb$ is a $P_4$ or a $C_4$.
Note that every obstruction contains exactly two diagonals and any solution must contain at least one of the two diagonals of each obstruction.

We first present a reduction rule that should be applied exhaustively,
and then two reduction rules that should be applied once.

\begin{reductionrule}\label{tp-reduce}
	Let $u,v$ be two non-adjacent vertices. If the number of obstructions containing $u,v$ is at least $k+1$, then add $uv$ to $E$ and decrease $k$ by 1.
\end{reductionrule}

\begin{restatable}[]{lemma}{tpreduce}
	Rule~\ref{tp-reduce} is safe.
\end{restatable}

Moreover, Rule~\ref{tp-reduce} can easily be applied in polynomial time.

The \emph{modulator $X(G)$} of $G$ is the subset of vertices of $G$ that are in at least one obstruction.
Guo~\cite[Theorem 4]{guo2007problem} stated that $(G,k)$ is a positive instance if and only if $(G[X(G)], k)$ is.
Therefore, the following reduction rule is safe:

\begin{reductionrule}\label{tp-modulator}
	If $X(G) \neq V(G)$, remove all vertices of $G$ that are not in $X(G)$.
\end{reductionrule}

When the first two rules cannot be applied, we perform the following rule which detects trivially negative instances.
\begin{reductionrule}\label{tp-size}
	If $|V(G)| > 2k^2+2k$, output a trivially negative instance.
\end{reductionrule}

To complete our proof, we simply have to prove that after applying the first two rules exhaustively, 
the size of a positive instance is quadratic. 
The next lemma ensures that Rule~\ref{tp-size} is safe, which concludes the proof of Theorem~\ref{thm:tpadd}.

\begin{restatable}[]{lemma}{tpbound}
	If $(G,k)$ is a positive instance and every diagonal belongs to at most $k$ obstructions,
	then $|X(G)| \leq 2k^2+2k$.
\end{restatable}

\section{Linear kernel for Starforest deletion}\label{sec:star}

%Recall that stars are trees with at most one internal vertex, and
%starforests are graphs that are disjoint unions of stars. 
The goal of this section is to prove Theorem~\ref{thm:star}.
Stars can be divided into two sets: centers and leaves. Let us define the notion of \emph{center set} of a star forest.
%Each star contains exactly one center. 
%The set of centers of a starforest forms an independent set and is a dominating set of the graph. Moreover each leaf has a unique neighbor which is the center of its star. 
%Intuitively, one would like to define leaves as degree one vertices and centers as the other vertices,
%but this definition yields multiple centers in $2$-stars. 
%Let us give a formal definition of centers of a starforest.

\begin{definition}[Center set]
	Let $\Sbb$ be a star-forest. A set $C^* \subseteq V(\Sbb)$
	is a \emph{center set} of $\Sbb$ if $C^*$ is a dominating set of $\Sbb$
	such that every star $S$ of $\Sbb$ contains exactly one vertex $c$ of $C^*$. 
	This vertex is called the center of $S$.
\end{definition}

Note that a center set is not necessarily unique since, in 2-stars, both vertices can be selected as a center.
Given a star forest $\Sbb$ with a set of centers $C^*$, the \emph{leaves} of $\Sbb$ are the vertices outside of $C^*$. By definition, every leaf has degree $1$ and its unique neighbor is in~$C^*$.

In what follows, we show how to use the structure of the input graph to identify and label vertices that are centers of an optimal solution, which leads to a \lnr{} kernelization algorithm.

Let $(G,k)$ be an instance of \stardel. 
Our first reduction rule, which is indeed safe, removes trivial connected components.

\begin{reductionrule}[Clean-up rule]\label{star-red-1}
	Remove from $G$ any connected component with $1$ or $2$ vertices.
\end{reductionrule}

Assume now that Rule~\ref{star-red-1} cannot be applied anymore.

%Let us first claim a couple of observations (formalized and proven in Section~\ref{star-safeness}).
%We will then use them to state two additional rules.
%Rule~\ref{star-center} labels in a subset $C$ of vertices,
%which are centers in an optimal modification of $G$.
%Then, Rules~\ref{star-center} and~\ref{star-center-reduce} use this labeling to reduce the instance, based on the following observations
%(formalized and proven in Section~\ref{star-safeness}):

%\begin{itemize}
%	\item 
%	if an unlabeled vertex $v$ is adjacent to multiple labeled vertices $c_1,\ldots,c_t$,
%	then the number of deletions needed to make $v$ adjacent to only $c_i$
%	is the same for every $i=1,\ldots,t$, 
%	\item 
%	if there is a solution, then there is an optimal solution in which all edges between labeled vertices 
%	are deleted,
%	\item if we merge all the labeled vertices together, then the minimum number of edges to delete	to obtain a s%tarforest is unchanged.
%\end{itemize}

\begin{reductionrule}[Center labeling rule]\label{star-center}
	Let $C$ be the set of vertices of $G$ that are adjacent to a vertex of degree $1$ in $G$.
	\begin{alphaenumerate}
		\item\label{star-center-r1} For every $v\notin C$, if $v$ is adjacent to a vertex $u$ of $C$, delete all the other edges between $v$ and $C$, and decrease the parameter accordingly. 
		\item\label{star-center-r2} For every $u,v\in C$, if $u$ and $v$ are adjacent then remove $uv$ from $G$ and decrease $k$ by $1$.
	\end{alphaenumerate}	
\end{reductionrule}

The fact that Rule~\ref{star-center} is safe is a consequence of the following lemma:
\begin{restatable}[]{lemma}{starcenters}\label{lemma:star-centers}
	Let $C$ be the set of vertices of $G$ that are adjacent to a vertex of degree $1$.
	If $(G,k)$ is a positive instance,
	then there exists a solution $F$ of $(G,k)$ and a center-set
	$C^*$ of $G-F$ such that $C \subseteq C^*$.
\end{restatable}
      
Lemma~\ref{lemma:star-centers} ensures that Rule~\ref{star-center} is safe. Indeed, if there exists a solution, then there is also a solution where $C$ is in the center-set.
Hence we can safely remove all the edges between the vertices of $C$ since each star only contains one vertex of the center-set of $G-F$.
Moreover, if an edge between $v$ and a vertex $w$ of $C$ is kept in $G-F$, 
then we can choose to keep any other edge between $v$ and $C$ instead of $vw$,
since all the vertices of $C$ are centers of their stars.

When neither Rule~\ref{star-red-1} nor Rule~\ref{star-center} can be applied, we apply the following rule:

\begin{reductionrule}[Center reduction rule]\label{star-center-reduce}
	Merge all the vertices of $C$, and remove all but $k+2$ vertices of degree 1. 
\end{reductionrule}

\begin{restatable}[]{lemma}{starcenterreduce}
  Rule~\ref{star-center-reduce} is safe. 
\end{restatable}

When Rules~\ref{star-red-1} to~\ref{star-center-reduce} cannot be applied, we apply the following rule once.

\begin{reductionrule}[Kernel size rule]\label{star-reject}
	If $|V(G)| > 4k+3$, return a trivial negative instance (e.g. $(P_4, 0)$).
	Otherwise, return $(G,k)$.
\end{reductionrule}

Rule~\ref{star-reject} ensures that the returned kernel has at most $4k+3$ vertices. 
In the remainder of this section, we study the structure of positive instances of \stardel{} 
to prove that Rule~\ref{star-reject} is safe.

In the two following lemmas, we assume that none of Rules~\ref{star-red-1} to~\ref{star-center-reduce} can be applied.
The following lemma uses the sparsity of starforests (they have many vertices of degree 1) to get information on the structure of positive instances.

\begin{restatable}[]{lemma}{manydegone}\label{lemma:many_deg_one}
	If $(G,k)$ is a positive instance of \stardel{} with $m$ edges, then $G$ contains at least $m-3k$ vertices of degree $1$.
\end{restatable}

In the last step of Rule~\ref{star-center}, we remove all but $k+2$ vertices of degree 1.
In the following lemma, we apply Lemma~\ref{lemma:many_deg_one} to show that the number of remaining vertices must be small.

\begin{restatable}[]{lemma}{starsmall}\label{lemma:star-small}
	If $(G,k)$ is a positive instance where no rule can be applied, then 
	$|V(G)| \leq 4k+3$.
\end{restatable}

By applying the contrapositive of Lemma~\ref{lemma:star-small},
we get that Rule~\ref{star-reject} is safe.

\subparagraph*{Improving the multiplicative constant in the linear bound.}
In the proof of Lemma~\ref{lemma:star-small}, we use a simple argument based on the minimum degree to show that the $3k$ remaining edges span at most $3k+1$ vertices. 
The worst case is when every vertex has degree $2$, that is, when every connected component is a cycle.
In a cycle, an optimal solution can easily be found in polynomial time, and therefore we can
remove cycles.
We can also show that long induced paths can be reduced. 
Combining these results gives a smaller kernel, at the cost of an increased running time and a slightly more involved analysis. 

However, these improvements do not yield a sublinear kernel.
It turns out that, under the Exponential Time Hypothesis, \stardel{} does not have a sublinear kernel.
Indeed, Drange et al.~\cite{drange2015fast} proved that, under ETH, \stardel{} does not admit a subexponential FPT algorithm,
i.e. an algorithm running in time $O^*(2^{o(k)})$.
Moreover, there is an $O^*(2^{n})$ algorithm for \stardel{}: for each subset $S$ of vertices, 
test whether there exists a solution in which $S$ is the center set.
Therefore, a kernel with $o(k)$ vertices would imply an $O^*(2^{o(k)})$ algorithm; a contradiction.

\bibliography{biblio.bib}

\newpage
\appendix
\section{Proofs for Section~\ref{sec:cliqueis}}

\cliqueisolfirst*

\begin{proof}
Note that removing an edge from $G$ reduces by one the degree of two
vertices.  Therefore, if $\delta(G) \geq k/(2\log k)$, then for every
solution $F$, the independent set part $I$ in $G-F$ contains at most
$4\log k$ vertices, since $F$ contains at most $k$ edges.  Moreover,
Damaschke and Mogren showed in~\cite[Lemma 5]{damaschke2014editing}
that a set $F$ is an \emph{optimal solution} of $G$ 
if and only if $I$ is a minimum
vertex cover of $\bar{G}$, where $I$ is
the independent set part of $G-F$.  In our case, this means that
the minimum vertex cover of $\bar{G}$ has size at most
$4 \log k$. 
As shown below, we can find such a minimum vertex cover 
in polynomial time in our case 
by combining an approximation algorithm with an exhaustive search on the
approximate solution.

Using a greedy 2-approximation algorithm for \vc{}, we can compute a
vertex cover $J$ of $\bar{G}$ such that $I\subseteq J$ and
$|J| \leq 2\cdot |I|$.  If $|J| > 8\log k$ then $|I| > 4\log k$, hence
the instance is negative. Otherwise, we can check if one of the $k^8$
subsets of $J$ is the set of isolated vertices of a
solution. Rule~\ref{clique-highdensity} can therefore be performed in
polynomial time. (This $k^8$ bound can actually be improved to
$k^{4(1+\varepsilon)}$ by using a PTAS for \vc{}~instead.) 
\end{proof}

\cliqueisolsecond*

\begin{proof}
  Observe that due to Rules~\ref{clique-lowdegree}
  and~\ref{clique-highdensity}, we have
  $\sqrt{2(m-k)}-1\leqslant \delta(G) < \frac{k}{2\log k}$, which
  implies that $m\leqslant \frac{k^2}{2\log^2 k}$ when $k> 257$. Note
  that we can assume w.l.o.g. that this last condition is satisfied
  since otherwise, we can solve the instance in polynomial time by
  taking every subset of $k$ edges as a candidate for $F$; this can be
  done in time $O(n^k)= O(n^{257})$ (note that this exponent can be reduced at the cost of
  increasing the multiplicative constant of our kernel).

  Similarly to Lemma~\ref{lemma:lb-clique}, we can show that the size
  of the clique in $G-F$ is at most $\sqrt{2m}+1 = k/\log k + 1$.
			
  Furthermore, Rule~\ref{clique-logdegree} ensures that
  $\delta(G)\geqslant 2\log k$. Removing an edge from $G$ reduces by
  one the degree of two vertices.  Therefore, by removing $k$ edges
  from $G$, we can make at most $k/\log k$ vertices isolated.
  Since all the vertices of $G$
  are either in the clique or isolated, $G$ has at most
  $k/\log k + 1 + k/\log k = 2k/\log k + 1$ vertices.
\end{proof}

\section{Proofs for Section~\ref{sec:split}}
\splitfirst*

\begin{proof}
  Let $(G',k',P')$ be the instance obtained by applying one of the
  rules to an instance $(G,k,P)$. We prove that $(G,k,P)$ is a
  positive generalized instance if and only if $(G',k',P')$ is. We
  first prove the direct implication.  By hypothesis, there exists a
  solution $F$ of $(G,k)$ such that $P \vDash F$; and let $(K^*,I^*)$
  be a split decomposition of $G+F$ witnessing that $P$ is compatible
  with $F$. We write $P'=(K',I',D')$ and $P=(K,I,D)$ and we consider
  three cases depending on the type of the applied rule.
	\begin{itemize}

	\item Rule~\ref{split-is}. Assume that 
	some vertex $v\in D$ is moved to $I'$.  If $v\in I^*$, then
	$I' = I\cup \{v\} \subseteq I^*$ and $K' = K \subseteq K^*$,
	hence $P' \vDash F$. Therefore, we can assume that
	$v\in K^*$.

	If $v$ is moved to $I'$ in application of Rule~\ref{split-is}-\ref{split-is-r1}, 
	then, for every neighbor $w$ of $v$ in $G+F$, either $w \in K$ or $vw \in F$.
	Removing the edges from $F$ incident to $v$
	provides a smaller solution $F'$ such that $G'+F'$ is a split graph, and all the neighbors of $v$ in $G'+F'$
	are in the clique. 
	Therefore, $G'$ has a split decomposition where $v$ is in the independent set, i.e. $P'\vDash F'$.
	
	If $v$ is moved to $I'$ in application of Rule~\ref{split-is}-\ref{split-is-r2},
	then $v$ is non-adjacent to at least $k+1$ vertices of
	$K \subseteq K^*$. Since $K^*$ is a clique in $G+F$,
	all the edges between $v$ and $K$ must be added and then
	$F$ must contain at
	least $k+1$ edges. This is a contradiction, as $|F| \leq k$.

	\item Rule~\ref{split-k}. Assume
	that %Rule~\ref{split-k} was applied, i.e.
	some vertex $v\in D$ is moved to $K'$. If $v\in K^*$, then
	$K' = K\cup \{v\} \subseteq K^*$ and $I' = I \subseteq I^*$,
	hence $P' \vDash F$. Therefore, we can assume that $v\in I^*$.

	If $v$ is moved to $K'$ in application of Rule~\ref{split-k}-\ref{split-k-r1}, then
	$v$ has a neighbor $u$ which belongs to $I \subseteq I^*$. Therefore,
	there is an edge $uv$ between two vertices that must belong to the independent set. 
	As we can only add edges, this is a contradiction.
	
	If $v$ is moved to $K'$ in application of Rule~\ref{split-k}-\ref{split-k-r2}, then
	there are at least $k+1$ non-edges in the graph induced by $N(v)$. Since $v\in I^*$ 
	and $I^*$ is an independent set, $N(v) \subseteq K^*$ must be in $K^*$.
	Since $K^*$ is a clique in $G+F$, 
	$F$ must contain the $k+1$ missing edges
	of $N(v)$, a contradiction with $|F| \leq k$.

	If $v$ is moved to $K'$ in application of Rule~\ref{split-k}-\ref{split-k-r3}, 
	then $v$ is adjacent to
	$K \cup D \supseteq K^*$. Observe that
	$(K^* \cup \{v\}, I^* \setminus \{v\})$ is a split decomposition of
	$G'+F$.	Moreover, since $K \subseteq K^*$ and
	$K' = K \cup \{v\}$, $K'$ is a subset of $K^* \cup \{v\}$, hence
	$P' \vDash F$.	Therefore, $F$ is a solution of $(G',k')$ and
	$P' \vDash F$, as desired.

	\item Rule~\ref{split-red}. Assume
	that Rule~\ref{split-red}-\ref{split-red-r1} was applied,
	i.e. some edge $e$ has been added (the other cases are trivial since the 
	instance $(G,k,P)$ cannot be positive). No vertex has moved, therefore
	we have $P' = P$, and $k' = k - 1$. Since $K\subset K^*$ and
	the endpoints of $e$ lie in $K$, $e$ is an edge in $G+F$ but
	not in $G$, hence $e\in F$.  Let $F' = F \setminus \{e\}$. By
	construction, $F'$ is a solution of $(G', k')$. Since
	$G + F = G' + F'$, $(K^*,I^*)$ is a split decomposition of
	$G'+F'$, and therefore $P' \vDash F'$.%\qedhere
	\end{itemize}

	We now prove the converse: assume that $(G',k',P')$ is positive.	
	In the case of Rule~\ref{split-is} and Rule~\ref{split-k},
	remark that any solution of $(G', k')$ is a solution of $(G, k)$ since the instance is unchanged.
	Hence, if $P'$ is compatible with a solution $F'$ of $(G',k')$,
	then $P$ is also compatible with $F'$, as $K\subseteq K'$ and $I \subseteq I'$.
	Therefore, we have a solution $F'$ of $(G,k)$ that $P$ respects, i.e. $(G,k,P)$ is positive.

	The case of Rule~\ref{split-red} follows by reversing the construction done in 
	the same case of the other direction.
\end{proof}

\splitsecond*

\begin{proof}
	Let $(G,k,P)$
	be the generalized instance before Rule~\ref{split-unlabel} is applied, and let $(G',k')$ be
	the instance that it returns. Let us prove that there exists a solution $F$ of	$(G,k)$ such that $P\vDash F$ if and only if $(G',k')$ has a solution.

	First, assume that there exists a solution $F$ of $(G,k)$ such that $P\vDash F$.
	Let $(K^*, I^*)$ be a split decomposition of $G+F$ that witnesses the fact
	that $P\vDash F$. Let $K^D = K^* \cap D$ and $I^D = I^*\cap D$.
	We can extract from $F$ a solution $F'$ of $(G, k)$ 
	that only adds edges between vertices of $K^*$.
	Notice that, by construction, $K^* = K \sqcup K^D$. 

	Recall that $(G',k')$ is obtained by replacing the subset of vertices $K$ in $G$ with 
	a set $K'$ of $k$ vertices.
	Moreover, by construction, 
	vertices of $K^D$ that are non-adjacent to $t$ vertices of $K$ in $G$
	are non-adjacent to $t$ vertices of $K'$ in $G'$.
	Rule~\ref{split-is}-\ref{split-is-r2} ensures that $t\leq k$, and
	Rule~\ref{split-red}-\ref{split-red-r1} ensures that $K$ induces a clique in $G$.
	Therefore, since $K'$ induces a clique in $G'$,
	the number of edges needed to turn $K^D \cup K$ into a clique 
	is the same as the number of edges needed to turn $K^D \cup K'$ into a clique.

	For every $u\in K^D$, let $uv_1,\ldots, uv_t$ be all the edges
    in $F$ that are adjacent to $u$ and some vertex $v_i \in K$.
    We construct a solution $F'$ of $(G',k')$ by replacing every
    such edge $uv_i$ of $F$ by $uv_i'$, where $v_i'\in K'$ is
    defined in Rule~\ref{split-red}-\ref{split-red-r2}. Since
    $k' = k$, $(G',k')$ is also a positive instance.
\smallskip

	Conversely, assume that $(G',k')$ has a solution $F'$.
	Since every vertex of $K'$ is adjacent to every vertex of $I'$,
	for any split decomposition of $(K^*,I^*)$ of $G'+F'$, we have
	$K'\subseteq K^*$. Otherwise, if some vertex of $K'$ were in $I^*$,
	$I'$ would be contained in the clique of $G'+F'$,
	which would imply that $F'$ contains more than $k'$ edges, a contradiction.
	Indeed, $I'$ is an independent set of size $\sqrt{2k'}$, and therefore, one requires more than $k'$ edge additions
	to turn it into a clique.
	As vertices of $I'$ are only adjacent to vertices of $K'$, 
	up to removing some edges from $F'$,
	we can assume that $I' \subseteq I^*$.

	Hence, by performing the same process as in the other direction in reverse,
	we can construct a solution $F$ of $(G,k)$ such that $G+F$
	admits a clique decomposition $(A,B)$ that satisfies $K\subseteq A$
	and $I\subseteq B$. In other words, we have $P\vDash F$, which concludes the proof.
\end{proof}

\firstobs*

\begin{proof}
	If $|I_v| > \sqrt{2k}+1$, 
	then $N(v)$ contains more than $\binom{\sqrt{2k}+1}{2} = k + \sqrt{k/2}$ non-edges.
	Hence, we can apply Rule~\ref{split-k}-\ref{split-k-r2}, a contradiction. 
\end{proof}

\secondobs*

\begin{proof}
	A vertex $v\in I^D$ can only have neighbors in $K$ and $K^D$.
	If $v$ does not have neighbors in $K^D$, it only has neighbors in $K$, 
	hence we can apply Rule~\ref{split-is}-\ref{split-is-r1}, a contradiction.
\end{proof}

\section{Proofs for Section~\ref{sec:trivially}}
\tpreduce*

\begin{proof}
If there is no solution, then indeed, the reduced instance still has no solution. Assume now that $(G,k)$ is positive. Let $F$ be a solution of $(G,k)$. If $F$ contains $uv$ the conclusion follows. Assume by contradiction that $uv$ is not in $F$. Since every obstruction contains exactly two diagonals, for every obstruction containing $u,v$ as a diagonal, the other diagonal is in $F$. Since the other diagonal consists of the other pair of vertices of the obstruction and there are at least $k+1$ disjoint pairs, $F$ must contain at least $k+1$ edges, a contradiction.
\end{proof}

\tpbound*

\begin{proof}
	Since $(G,k)$ is a positive instance, there exists a solution $F$ containing at most $k$ edges. 
	Moreover every edge $e$ of $F$ belongs to at most $k$ obstructions in $G$. 
	Therefore, the number of vertices that are in an obstruction containing $e$ as a diagonal is at most $2k+2$ (since all the obstructions contain the endpoints of $e$). 
	
	Let $x$ be a vertex of $X(G)$. 
	Since all the vertices of $X(G)$ belong to at least one obstruction of $G$ and $F$ is a solution, 
	$F$ contains at least one of the two diagonals of some obstruction containing $x$. 
	Therefore, we can map each vertex of $X(G)$ to an edge of $F$ that is a diagonal of an obstruction containing $x$. 
	The first part of the proof ensures that the number of vertices mapped to any edge of $F$ is at most $2k+2$, 
	therefore the total size of $X(G)$ is at most $2k^2+2k$, which completes the proof.
\end{proof}

\section{Proofs for Section~\ref{sec:star}}
\starcenters*

\begin{proof}
	Let $F$ be a solution of $(G,k)$ that maximizes the number of vertices of $C$ in a center-set $C^*$ of $G-F$.
	Let $\Sbb = G-F$. Assume by contradiction that $C$ is not included in $C^*$. 
%	Let $D$ be the set of vertices with a degree-$1$ neighbor in $G$ 
%	and that are not in $C^*$. 
	We prove that there exists a solution $F'$ of $(G,k)$ and a center set $C'$ of $G-F'$ 
	containing more vertices from $C$ than $C^*$, a contradiction.% such that the set $D'$ of vertices of $C$ that are in $C'$ is strictly smaller than $D$, 
%	contradicting the minimality of $D$.
	
	Let $v \in C\setminus C^*$. Let $w$ be a neighbor of $v$ of degree 1 in $G$. Observe that
	$v$ cannot have degree 1 in $G$, otherwise $\{v,w\}$ induces a $2$-star in $G$. Therefore,
	$v$ has degree at least $2$ in $G$.	However, it has degree 1 in $G-F$ since $v\notin C^*$.
	Notice that $w\in C^*$, as it is either isolated or adjacent to $v$ in $G-F$, and $v\notin C^*$.

	If $w$ is the center of the star containing $v$, that star is reduced to $\{v,w\}$.
	Hence, we can replace
	$w$ with $v$ in $C^*$, which increases the size of $C\cap C^*$, a contradiction.

	Otherwise, let $u \ne w$ be the center of the star of $v$.	We
	set $F'= (F \setminus vw ) \cup uv$ (see Figure~\ref{fig:star-swap}). We have $|F'| \le |F|$
	and $F'$ is still a solution of $(G,k)$.	Moreover, $G-F'$ has a center set $C'$ containing
	$C^* \cap C$ and $v$, which contradicts the maximality of $F$.\qedhere

	\begin{figure}[htbp]
	\begin{center}
	\begin{tikzpicture}
		\node[draw, circle] (v) at (0.5,0)	 {$v$};
		\node[draw, circle] (w) at (2,0)	 {$w$};
		\node[draw, circle] (u) at (-1,0)	 {$u$};
		\node[draw, circle] (a) at (-2,1) {};
		\node[draw, circle] (b) at (-2,0) {};
		\node[draw, circle] (c) at (-2,-1) {};
		\node[] (g) at (0,1) {In $G$};

		\draw[thick] (a) -- (u) -- (b);
		\draw[thick] (c) -- (u);
		\draw[thick] (u) -- (v) -- (w);
	\end{tikzpicture}%

	\begin{tikzpicture}
		\node[draw, circle] (v) at (0.5,0)	 {$v$};
		\node[draw, circle, red] (w) at (2,0)	 {$w$};
		\node[draw, circle, red] (u) at (-1,0)	 {$u$};
		\node[draw, circle] (a) at (-2,1) {};
		\node[draw, circle] (b) at (-2,0) {};
		\node[draw, circle] (c) at (-2,-1) {};

		\draw[thick] (a) -- (u) -- (b);
		\draw[thick] (c) -- (u);
		\draw[thick] (u) -- (v);
		\node[] (g) at (0,-1) {In $G-F$};
	\end{tikzpicture}%
	\hspace{2cm}%
	\begin{tikzpicture}
		\node[draw, circle, red] (v) at (0.5,0)	 {$v$};
		\node[draw, circle] (w) at (2,0)	 {$w$};
		\node[draw, circle, red] (u) at (-1,0)	 {$u$};
		\node[draw, circle] (a) at (-2,1) {};
		\node[draw, circle] (b) at (-2,0) {};
		\node[draw, circle] (c) at (-2,-1) {};

		\draw[thick] (a) -- (u) -- (b);
		\draw[thick] (c) -- (u);
		\draw[thick] (v) -- (w);
		\node[] (g) at (0,-1) {In $G-F'$};
	\end{tikzpicture}%
	\caption{Illustration of the transformation used to make $v$ a center, where $v$ is a vertex adjacent to the degree-1 vertex $w$.
	Vertices of the center set ($C^*$ in $G-F$, $C'$ in $G-F'$) are drawn in red.}
	\label{fig:star-swap}
	\end{center}
	\end{figure}
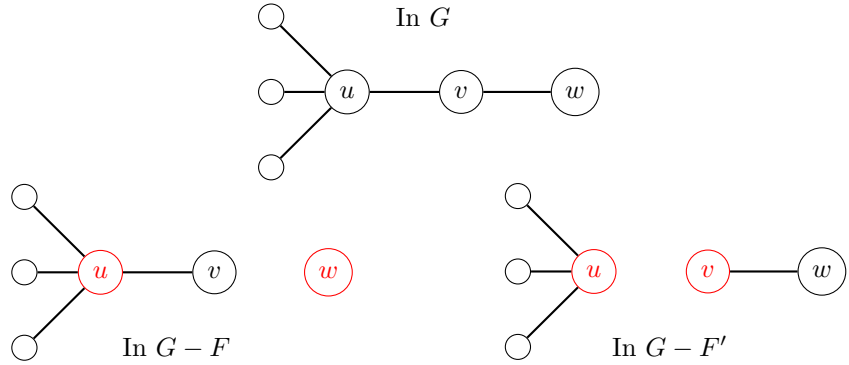
\end{proof}

\starcenterreduce*

\begin{proof}
	Assume that Rule~\ref{star-center} cannot be applied.
	Let $(G_0,k)$ be the instance before Rule~\ref{star-center-reduce} is applied,
	let $(G_1,k)$ be the instance after vertices of $C$ have been merged to the vertex $c$,
	and let $(G_2,k)$ be the instance that Rule~\ref{star-center-reduce} returns.
	Since Rule~\ref{star-center} cannot be applied, there is no edge between vertices of $C$, and each vertex of $G_0$ is adjacent to at most one vertex of $C$. Therefore there are no loops nor parallel edges in $G_1$. 

	The instances $(G_0,k)$ and $(G_1,k)$ are equivalent.
	Indeed, by Lemma~\ref{lemma:star-centers}, if $G_0$ (resp. $G_1$) is positive, 
	there exists a solution $F$ (resp. $F'$) such that $C$ (resp. $\{c \}$) is in a center-set of $G_0-F$ (resp. $G_1-F'$).
	For every vertex $v$ of degree 1 in $G_0$, let $c_v$ be its only neighbor (which is in $C$).
	We can then transform a solution of $F$ of one instance, with a center set that contains $C$ or $c$, into a solution of the other instance by swapping edges of the form $(v,c)$ for edges $(v,c_v)$. This operation does not change the cardinality of the solution, and the instances have the same parameter, therefore the instances are equivalent.
	\medskip

	Let us finally prove that	$(G_1,k)$ is positive if and only if $(G_2,k)$ is. 
	Since $G_2$ is a subgraph of $G_1$, if $(G_1,k)$ is positive then removing the same set of edges in $G_2$ also gives a solution.

	Now, assume that $(G_2,k)$ is positive. Let $F$ be a solution of $(G_2,k)$. 
	W.l.o.g., $c$ is adjacent to $k+2$ leaves in $G_1$ (otherwise $G_1=G_2$, and the equivalence is trivial). 
	As $|F|\le k$, $c$ is still adjacent to at least $2$ leaves in $G_2-F$, and therefore, $c$ is the center of its star. 
	This implies that the same set of edges $F$ is a solution of $G_1$ since $G_1$ is obtained from $G_2$ by adding pendant vertices adjacent to $c$, and $c$ is a center of any solution of $G_2$.
\end{proof}

\manydegone*
\begin{proof}
	Since $(G,k)$ is a positive instance, there exists a set $F \subseteq E$ of size at most $k$ 
	and a starforest $\Sbb$ such that $\Sbb = G - F$. 
	Let $t$ denote the number of stars (i.e. of connected components) in $\Sbb$.
	In each star of $\Sbb$, there is at most one vertex which does not have degree 1.
	Therefore, $\Sbb$ has at least $n-t$ vertices of degree 1.
	Let $\ell$ be the number of vertices in $G$ that have degree 1.
	We can obtain $G$ from $\Sbb$ by adding at most $k$ edges, and
	adding an edge can change the degree of at most two vertices. 
	Therefore, we have $\ell \ge n - t - 2k$.
	Moreover, $\Sbb$ has $n-t$ edges, hence $m \leq n-t+k$.
	By combining the above, we obtain that $\ell \geq n - 2k + m - n -k = m - 3k$.
\end{proof}

\starsmall*
\begin{proof}
	By Lemma~\ref{lemma:many_deg_one}, $G$ contains at least $m-3k$ vertices of degree 1. 
	Since Rule~\ref{star-red-1} cannot be applied, $G$ does not contain 2-stars, and therefore each edge is incident to at most one vertex of degree 1. 
	Hence, removing all vertices of degree 1 from $G$ removes the same number of edges,
	i.e. we obtain a graph $H$ with at most $3k$ edges.
	As all the degree 1 vertices of $G$ adjacent to a single vertex $v$,
	all vertices of $H$ but $v$ have degree at least 2.
	Hence, $H$ contains at most $3k+1$ vertices.
	Moreover, by Rule~\ref{star-center-reduce} $G$ contains at most $k+2$ vertices of degree 1.
	Therefore, we have $|V(G)| \leq 3k +1 + k+2 = 4k+3$.
\end{proof}

\end{document}